\documentclass[a4paper,12pt]{article}

\usepackage{amsfonts,graphicx,subfig,amsmath,amssymb,enumitem,
,url,color,authblk,amsthm,multicol,cite,float}

\DeclareMathOperator\erf{erf}

\theoremstyle{definition}
\newtheorem{definition}{Definition}[section]
\theoremstyle{theorem}
\newtheorem{theorem}{Theorem}[section]
\theoremstyle{lemma}
\newtheorem{lemma}{Lemma}[section]
\theoremstyle{corollary}

\theoremstyle{remark}
\newtheorem{remark}{Remark}

\title{Determination of two unknown thermal coefficients through an inverse one-phase fractional Stefan problem}

\author[1,2]{Andrea N. Ceretani\thanks{aceretani@austral.edu.ar}}
\author[1]{Domingo A. Tarzia\thanks{dtarzia@austral.edu.ar}}
\affil[1]{{\small CONICET - Depto. Matem\'atica, Facultad de Ciencias Empresariales, Universidad Austral, Paraguay 1950, S2000FZF Rosario, Argentina.}}
\affil[2]{{\small Depto. de Matem\'atica, Facultad de Ciencias Exactas, Ingenier\'ia y Agrimensura, Universidad Nacional de Rosario, Pellegrini 250, S2000BTP Rosario, Argentina.}}
\date{}

\begin{document}  
\maketitle

\begin{abstract}
We consider a semi-infinite one-dimensional phase-change material with two unknown constant thermal coefficients among the latent heat per unit mass, the specific heat, the mass density and the thermal conductivity. Aiming at the determination of them, we consider an inverse one-phase Stefan problem with an over-specified condition at the fixed boundary and a known evolution for the moving boundary. We assume that the phase-change process presents latent-heat memory effects by considering a fractional time derivative of order $\alpha$ ($0<\alpha<1$) in the Caputo sense and a sharp front model for the interface. According to the choice of the unknown thermal coefficients, six inverse fractional Stefan problems arise. For each of them, we determine necessary and sufficient conditions on data to obtain the existence and uniqueness of a solution of similarity type. Moreover, we present explicit expressions for the temperature and the unknown thermal coefficients. Finally, we show that the results for the classical statement of this problem, associated with $\alpha=1$, are obtained through the fractional model when $\alpha\to 1^-$.
\end{abstract}



\section{Introduction}
Determination of thermal coefficients for phase-change materials through inverse Stefan problems has been widely studied during the last decades \cite{CeTa2015,CeTa2016,KaIs2012,Ta1982,Ta1983,Ta1987}. Especially, phase-change processes involving solidification or melting have been extensively studied because of their scientific and technological applications \cite{AlSo1993,Ca1984,CaJa1959,Cr1984,Fa2005,Gu2003,Lu1991,Ru1971,Ta2011}. A rewiew of a long bibliography on moving and free boundary value problems for the heat equation can be consulted in \cite{Ta2000}. Recently, a new sort of Stefan problems including time-fracional derivatives have begun to be studied \cite{JiMi2009,KhZaFe2003,Vo2010,Vo2014,VoFaGa2013,RoSa2013,
RoSa2014,RoTa2014}. Some references in fractional derivatives can be found at \cite{KiSrTr2006,Ma2010,Po1999,GoLuMa1999,Lu2010,MaLuPa2001}. In \cite{RoSa2013,RoSa2014,RoTa2014} there are considered free boundary value problems which are obtained by replacing the time derivative in a one-phase Stefan problem by a fractional derivative of order $\alpha$ ($0<\alpha<1$) in the Caputo sense \cite{Ca1967}, and explicit solutions of similarity type are given for the resultant {\em fractional Stefan problems}. A physical interpretation of the problems considered in \cite{RoSa2013,RoSa2014,RoTa2014} is given in \cite{VoFaGa2013}. In that article, authors derive fractional Stefan problems for phase-change processes by substituting the {\em local} expression of the heat flux given by the Fourier law for a new {\em non-local} definition. They consider a heat flux given as a weighted sum of local fluxes back in time, which they express in terms of the Riemann-Liouville integral of order $\alpha$ ($0<\alpha<1$) of the local flux given by the Fourier law. They also explain how this change implies that the new model takes into consideration latent-heat memory effects in the evolution of the phase-change process and give to the parameter $\alpha$ the physical meaning of being the {\em strength of memory retention}. This fractional model reduces to the classical Stefan problem when $\alpha=1$. The same occurs with the solutions of similarity type given in \cite{RoSa2013,RoSa2014,RoTa2014} in the sense that they converge to the similarity solutions of the classical Stefan problems with which they are related to, when $\alpha\to 1^-$. To the authors knowledge, the first use of inverse fractional Stefan problems for the determination of thermal coefficients has been done recently in \cite{Ta2015-c}. In that article the author studies the determination of one unknown thermal coefficient for a semi-infinite material through a fractional one-phase Stefan problem with an over-specified condition at the fixed boundary. Necessary and sufficient conditions on data to obtain the existence and uniqueness of solutions of similarity type are established, and explicit expressions for the the temperature of the material, the free boundary and the unknown thermal coefficient are given. Moreover, it has shown that results through the fractional model reduce to the results previously obtained in \cite{Ta1982} for the determination of one unknown thermal coefficient using a classical inverse Stefan problem. Encouraged by \cite{Ta1983,Ta2015-c}, we consider here the problem of determining two unknown thermal coefficients through an inverse fractional one-phase Stefan problem for which it is known the evolution of the free boundary. 
In order to have dimensional coherence in the time fractional heat equation as well as in the fractional Stefan condition, we have included two extra parameters $\mu_\alpha,\,\nu_\alpha\in(0,1]$ in the model, which are such that:
\begin{subequations}\label{munu-limit}
\begin{align}
\label{mu-limit}&\mu_\alpha\to 1\quad\quad\text{when }\alpha\to 1^-\\
\label{nu-limit}&\nu_\alpha\to 1\quad\quad\text{when }\alpha\to 1^-.
\end{align}
\end{subequations}
In particular, $\mu_\alpha$ and $\nu_\alpha$ can be considered equal to 1 with the corresponding physical dimension (see below). 

\noindent More precisely, we consider the following inverse problem for a one-phase melting process:
\begin{subequations}\label{Problema}
\begin{align}
\label{1}&D^\alpha T(x,t)=\mu_{\alpha}\lambda^2T_{xx}(x,t)&0<x<s(t),\,&t>0\\
\label{2}&T(s(t),t)=T_m& &t>0\\
\label{3}&-kT_x(s(t),t)=\nu_\alpha\rho lD^\alpha s(t)& &t>0\\
\label{4}&T(0,t)=T_0& &t>0\\
\label{5}&kT_x(0,t)=-\frac{q_0}{t^{\alpha/2}}& &t>0
\end{align}
\end{subequations}   
where the unknowns are the temperature $T$ [$^\circ$C] of the liquid phase and two thermal coefficients among:
\begin{table}[h!]
\begin{tabular}{rll}
$k>0$:& thermal conductivity& [$W\,{m^{-1}}\,\left(^\circ C\right)^{-1}$]\\
$\rho>0$:& mass density& [$kg\,m^3$]\\
$l>0$:& latent heat per unit mass& [$J\,kg^{-1}$]\\
$c>0$:& specific heat& [$J\,kg^{-1}\,\left(^\circ C\right)^{-1}$]
\end{tabular}
\end{table}

\noindent According to \cite{Hr2011-a,Hr2011-b}, the coefficient $\mu_\alpha\lambda^2$ [$m^2\,s^{-\alpha}$] in equation (\ref{1}) is a sort of {\em fractional diffusion coefficient}, $\lambda^2$ [$m^2\,s^{-1}$] being the thermal diffusivity given by:
\begin{equation*}
\lambda^2=\frac{k}{\rho c}\hspace{2cm}(\lambda>0).
\end{equation*} 

\noindent We assume that the remaining coefficients:
\begin{table}[h!]
\begin{tabular}{rll}
$T_m>0$:& phase-change temperature& [$^\circ C$]\\
$T_0>T_m$:& temperature at the boundary $x=0$& [$^\circ C$]\\
$q_0>0$:& coefficient characterizing the heat flux at $x=0$& [$W\,m^2\,\left(^\circ C\right)^{-1}$]\\
$0<\alpha<1$:& strength of the memory retention & (dimensionless)\\
$0<\mu_{\alpha}\leq 1$: & parameter required to have dimensional coherence in equation (\ref{1}) &[$s^{1-\alpha}$]\\
$0<\nu_{\alpha}\leq 1$: & parameter required to have dimensional coherence in condition (\ref{3}) &[$s^{\alpha-1}$]
\end{tabular}
\end{table}

\noindent  involved in Problem (\ref{Problema}), are all known (e.g. through a phase-change experiment). Aiming at the simultaneous determination of two unknown thermal coefficients, we consider that the time evolution of the sharp interface $s$ is also known. More precisely, we follow \cite{RoSa2013,RoSa2014,RoTa2014} in assuming that it is given by:
\begin{equation}\label{s}
s(t)=\sigma t^{\alpha/2},\quad t>0,
\end{equation}
with $\sigma>0$ [$m\,s^{-\alpha/2}$]. The operator $D^\alpha$ in (\ref{1}) and (\ref{3}) represents the fractional time derivative of order $\alpha$ in the Caputo sense, which is defined by \cite{Ca1967}:
\begin{equation}\label{DerCaputo}
D^\alpha f(t)=\left\{\begin{array}{lcl}
\frac{1}{\Gamma(1-\alpha)}\displaystyle\int_0^t\frac{f'(\tau)}{(t-\tau)^\alpha}d\tau&\text{ if }&0<\alpha<1\\
f'(t)&\text{ if }&\alpha=1
\end{array}\right.,\quad t>0
\end{equation} 
for any $f\in W^1(\mathbb{R}^+)=\left\{f\in C^1(\mathbb{R}^+)/f'\in L^1(\mathbb{R}^+)\right\}$, where $\Gamma$ is the Gamma function defined by:
\begin{equation*}
\Gamma(x)=\displaystyle\int_0^{+\infty} s^{x-1}\exp(-s)ds,\quad x>0.
\end{equation*}

\section{Solutions of similarity type}

We begin with a definition of being a solution to the inverse fractional Stefan problem (\ref{Problema}), which is based in the definition established in \cite{RoSa2013,RoSa2014} for the direct case.

\begin{definition}
The triplet given by the temperature $T$ and two unknown thermal coefficients among $k$, $\rho$, $l$ and $c$ is a solution to the inverse fractional one-phase Stefan problem (\ref{Problema}) if:
\begin{enumerate}\label{DefSol}
\item $T$ is defined in $\mathbb{R}_0^+\times\mathbb{R}_0^+$.
\item $T\in C(\Omega)\cap C_x^2(\mathbb{R}^+\times\mathbb{R}^+)\cap W_t^1(\mathbb{R}^+)$, 
where $\Omega=\left\{(x,t)\,/\,0<x<s(t),\,t>0\right\}$ and $W_t^1(\mathbb{R}^+)=\left\{f:\mathbb{R}^+\times\mathbb{R}^+\to\mathbb{R}/\,f(x,\cdot)\in W^1(\mathbb{R}^+)\,\forall\,x>0\right\}$.
\item $T$ is a continuous function over $\Omega\cup\partial_p\Omega$, where $\partial_p\Omega=\left\{(0,t)/t>0\right\}\cup\left\{(s(t),t)/t>0\right\}$, except maybe in the point $(0,0)$, in which it must be satisfied the following condition:
\begin{equation*}
0\leq \displaystyle\liminf_{(x,t)\to(0,0)}T(x,t)\leq\displaystyle\limsup_{(x,t)\to(0,0)}T(x,t)<\infty.
\end{equation*}
\item For all $t>0$, exists $\frac{\partial}{\partial x}T(s(t),t)$.
\item The unknown thermal coefficients are positive real numbers.
\item The temperature $T$ and the two unknown thermal coefficients verify (\ref{Problema}).
\end{enumerate}
\end{definition} 

\begin{remark}
We note that the function $s$ given by (\ref{s}) is consistent with the definition of being a solution to a (direct) fractional one-phase Stefan problem given in \cite{RoSa2013,RoSa2014}, since it is a positive function belonging to $C(\mathbb{R}_0^+)\cap W^1(\mathbb{R}^+)$. 
\end{remark}

Encouraged by \cite{RoSa2013,RoSa2014,RoTa2014,Ta2015-c}, in this section we are looking for a solution of similarity type to Problem (\ref{Problema}). That is, we look for a temperature function $T$ such that:
\begin{equation}\label{TPreliminar}
T(x,t)=A+B\left(1-W\left(-\frac{x}{\sqrt{\mu_\alpha}\lambda t^{\alpha/2}},-\frac{\alpha}{2},1\right)\right),\quad 0<x<s(t),\,t>0,
\end{equation}  
where $A$ and $B$ are real numbers that must be determined, and $W$ is the {\em Wright function} given by \cite{Wr1934,Wr1940}:
\begin{equation}\label{Wright}
W(z,a,b)=\displaystyle\sum_{k=0}^{+\infty}\frac{z^k}{k!\Gamma(a k+b)},\quad z\in\mathbb{C},\,a>-1,\,b\in\mathbb{C}.
\end{equation} 

According to \cite{RoSa2013,RoSa2014}, we know that the function $T$ given by (\ref{TPreliminar}) fulfill conditons 1 to 4 in Definition \ref{DefSol} and that it satisfies the fractional diffusion equation (\ref{1}).   Noting that $W\left(0,-\frac{\alpha}{2},1\right)=1$, it follows from condition (\ref{4}) that:
\begin{equation}\label{A}
A=T_0.
\end{equation} 
From this and condition (\ref{2}), we have that:
\begin{equation}\label{B}
B=\frac{T_m-T_0}{1-W\left(-\frac{\sigma}{\sqrt{\mu_\alpha}\lambda},-\frac{\alpha}{2},1\right)}.
\end{equation} 
We observe that $W\left(-\frac{\sigma}{\sqrt{\mu_{\alpha}}\lambda},-\frac{\alpha}{2},1\right)\neq 1$ because $W\left(-x,-\frac{\alpha}{2},1\right)$ is a strictly decreasing function in $\mathbb{R}^+$ \cite{RoSa2013} and, as we have already noted, $W\left(0,-\frac{\alpha}{2},1\right)=1$. 

\noindent By taking into consideration that the Wright function satisfies \cite{Wr1934,Wr1940}:
\begin{equation*}
\frac{d}{dz}W(z,a,b)=W(z,a,a+b),\quad z\in\mathbb{C},\,a>-1,\,b\in\mathbb{C},
\end{equation*}
and the fractional Caputo derivative of a power function with positive exponent is given by \cite{Po1999}:
\begin{equation*}
D^\alpha t^p=\frac{\Gamma(p+1)}{\Gamma\left(p-\alpha+1\right)}t^{p-\alpha},\quad t>0,\,p>0,
\end{equation*}
it follows from the fractional Stefan condition (\ref{3}) that:
\begin{equation}\label{Eq1-Preliminar}
\frac{\sqrt{\mu_\alpha}l\left[1-W\left(-\frac{\sigma}{\sqrt{\mu_{\alpha}}\lambda},-\frac{\alpha}{2},1\right)\right]}{\lambda c M_{\frac{\alpha}{2}}\left(\frac{\sigma}{\sqrt{\mu_\alpha}\lambda}\right)}=
\frac{(T_0-T_m)\Gamma\left(-\frac{\alpha}{2}+1\right)}{\nu_\alpha\sigma\Gamma\left(\frac{\alpha}{2}+1\right)},
\end{equation}
where $M_{\nu}$ is the {\em Mainardi function}, which is defined by \cite{Ma1995}:
\begin{equation}\label{Mainardi}
M_{\nu}(z)=W\left(-z,-\nu,1-\nu\right),\quad z\in\mathbb{C},\,0<\nu<1
\end{equation}
and satisfies $M_\nu(z)>0$ for all $z\in\mathbb{R}^+$ \cite{RoSa2013}.
 
\noindent Finally, when we consider $B$ given by (\ref{B}), the heat flux boundary condition (\ref{5}) implies that it must be satisfied the following equality:
\begin{equation}\label{Eq2-Preliminar}
\frac{\sqrt{\mu_{\alpha}}q_0\lambda}{k}\left[1-W\left(-\frac{\sigma}{\sqrt{\mu_{\alpha}}\lambda},-\frac{\alpha}{2},1\right)\right]=\frac{T_0-T_m}{\Gamma\left(-\frac{\alpha}{2}+1\right)}.
\end{equation}

We have thus proved the following result:

\begin{theorem}\label{ThCaracterizacion}
If the moving boundary $s$ is defined by (\ref{s}), then the function $T$ given by:
\begin{equation}\label{T}
T(x,t)=T_0-\frac{T_0-T_m}{1-W\left(-\xi,-\frac{\alpha}{2},1\right)}\left[1-W\left(-\frac{x}{\sqrt{\mu_\alpha}\lambda t^{\alpha/2}},-\frac{\alpha}{2},1\right)\right],\quad 0<x<s(t),\,t>0
\end{equation}
is a solution to Problem (\ref{Problema}) with two unknown thermal coefficients among $k$, $\rho$, $l$ and $c$, if and only if these are a solution to the following system of equations:
\begin{subequations}\label{SistEq}
\begin{align}
\label{Eq1}&\frac{\xi\left[1-W\left(-\xi,-\frac{\alpha}{2},1\right)\right]}{M_{\frac{\alpha}{2}}(\xi)}=
\frac{c(T_0-T_m)\Gamma\left(-\frac{\alpha}{2}+1\right)}{\mu_\alpha\nu_\alpha l\Gamma\left(\frac{\alpha}{2}+1\right)}\\
\label{Eq2}&1-W\left(-\xi,-\frac{\alpha}{2},1\right)=\frac{\sqrt{k\rho c}(T_0-T_m)}{\sqrt{\mu_{\alpha}}q_0\Gamma\left(-\frac{\alpha}{2}+1\right)}
\end{align}
\end{subequations}
where the dimensionless parameter $\xi$ is defined by:
\begin{equation}\label{xi}
\xi=\frac{\sigma}{\sqrt{\mu_\alpha}\lambda}.
\end{equation}
\end{theorem}

\section{Existence and uniqueness of solutions of similarity type. Formulae for the two unknown thermal coefficients.}\label{SeFormulas}
In this section we will look for necessary and sufficient conditions on data to have existence and uniqueness of solution to Problem (\ref{Problema}) for each possible choice of the two unknown thermal coefficients, as well as explicit formulae for them. Thanks to Theorem \ref{ThCaracterizacion}, it can be done through analysing and solving the system of equations (\ref{SistEq}) for each pair of unknown thermal coefficients. With the aim of organizing the main results of this section, we will write:
\begin{table}[h!]
\begin{tabular}{ll}
Case 1: Determination of $l$ and $c$&\hspace*{2cm}
Case 4: Determination of $c$ and $\rho$\\
Case 2: Determination of $c$ and $k$&\hspace*{2cm}
Case 5: Determination of $l$ and $\rho$\\
Case 3: Determination of $l$ and $k$&\hspace*{2cm}
Case 6: Determination of $\rho$ and $k$.
\end{tabular}
\end{table}

For each $\alpha\in(0,1)$, we introduce the real functions $F_\alpha$, $G_\alpha$ and $H_\alpha$ defined in $\mathbb{R}^+$ by:
\begin{subequations}\label{FGH}
\begin{align}
\label{F}&F_\alpha(x)=\frac{f_\alpha(x)}{x}\\
\label{G}&G_\alpha(x)=xf_\alpha(x)\\
\label{H}&H_\alpha(x)=\frac{xf_\alpha(x)}{M_{\frac{\alpha}{2}}(x)}
\end{align}
\end{subequations}
where
\begin{equation}\label{f}
f_\alpha(x)=1-W\left(-x,-\frac{\alpha}{2},1\right).
\end{equation}

The following result will be useful all throughout this section:

\begin{lemma}\label{LeFGHProp}
For any $\alpha\in(0,1)$, the real functions $F_\alpha$, $G_\alpha$ and $H_\alpha$ defined in (\ref{FGH}) satisfy the following conditions:
\begin{subequations}\label{FGHProp}
\begin{align}
\label{FProp}F_\alpha(0^+)&=\frac{1}{\Gamma\left(-\frac{\alpha}{2}+1\right)},
&F_\alpha(+\infty)&=0,
&F_\alpha'(x)&<0\quad\forall\,x>0\\
\label{GProp}G_\alpha(0^+)&=0,
&G_\alpha(+\infty)&=+\infty,
&G_\alpha'(x)&>0\quad\forall\,x>0\\
\label{HProp}H_\alpha(0^+)&=0,
&H_\alpha(+\infty)&=+\infty,
&H_\alpha'(x)&>0\quad\forall\,x>0.
\end{align}
\end{subequations}
\end{lemma}

\begin{proof}
The proof of (\ref{FProp}) was done in \cite{Ta2015-c}. The demonstrations of (\ref{GProp}) and (\ref{HProp}) follow from elementary computations and the following facts:
\begin{enumerate}
\item Since $0<\alpha<1$, $f_\alpha$ is a positive and strictly increasing function in $\mathbb{R}^+$ \cite{RoSa2013}.
\item Since $0<\alpha<1$, $M_{\frac{\alpha}{2}}$ is a positive and strictly decreasing function in $\mathbb{R}^+$ \cite{RoSa2013}.
\item $\displaystyle\lim_{x\to+\infty}f(x)=1$ and $\displaystyle\lim_{x\to+\infty}M_{\frac{\alpha}{2}}(x)=0$ \cite{GoLuMa1999}. 
\end{enumerate} 
\end{proof}
\begin{theorem}[Case 1: Determination of $l$ and $c$]\label{Thlc}
If the moving boundary $s$ is given by (\ref{s}), then the Problem (\ref{Problema}) admits the solution $T$, $l$ and $c$ given by (\ref{T}) and:
\begin{subequations}\label{lc}
\begin{align}
\label{c-lc}&c=\frac{1}{\rho k}\left[\frac{q_0\sqrt{\mu_\alpha}\Gamma\left(-\frac{\alpha}{2}+1\right)\left[1-W\left(-\xi,-\frac{\alpha}{2},1\right)\right]}{T_0-T_m}\right]^2\\
\label{l-lc}&l=\frac{q_0^2\,\Gamma^3\left(-\frac{\alpha}{2}+1\right)\left[1-W\left(-\xi,-\frac{\alpha}{2},1\right)\right]M_{\frac{\alpha}{2}}(\xi)}{\nu_\alpha\rho k(T_0-T_m)\Gamma\left(\frac{\alpha}{2}+1\right)\xi}
\end{align}
\end{subequations}
respectively, $\xi$ being the only one solution to the equation:
\begin{equation}\label{EqXi-lc}
F_\alpha(x)=\frac{k(T_0-T_m)}{\sigma q_0\Gamma\left(-\frac{\alpha}{2}+1\right)},\quad x>0,
\end{equation}
if and only if the following inequality holds:
\begin{equation}\label{Cond-lc}
\frac{k(T_0-T_m)}{\sigma q_0}<1.
\end{equation}
\end{theorem}

\begin{proof}
Isolating $c$ from equation (\ref{Eq2}), we have that $c$ is given by (\ref{c-lc}). Now, by combining this with equation (\ref{Eq1}), it can be obtained that $l$ is given by (\ref{l-lc}). It must be noted that the parameter $\xi$ involved in both (\ref{c-lc}) and (\ref{l-lc}) depends on $c$. Nevertheless, it can be determined without making any reference to $c$ as follows. By replacing (\ref{c-lc}) in the definition of $\xi$ given in (\ref{xi}), we have that $\xi$ must be a solution to equation (\ref{EqXi-lc}). It follows from (\ref{FProp}) that the equation (\ref{EqXi-lc}) admits a solution if and only if its RHS is between $0$ and $\frac{1}{\Gamma\left(-\frac{\alpha}{2}+1\right)}$. To complete the proof only remains to observe that this is equivalent to say that inequality (\ref{Cond-lc}) must hold and that, when this happens, equation (\ref{EqXi-lc}) has an only one positive solution.
\end{proof}

\begin{theorem}[Case 2: Determination of $c$ and $k$]\label{Thck}
If the moving boundary $s$ is given by (\ref{s}), then the Problem (\ref{Problema}) admits the solution $T$, $c$ and $k$ given by (\ref{T}) and:
\begin{subequations}\label{ck}
\begin{align}
\label{c-ck}&c=\frac{\mu_{\alpha}\nu_\alpha l\Gamma\left(\frac{\alpha}{2}+1\right)\xi\left[1-W\left(-\xi,-\frac{\alpha}{2},1\right)\right]}{(T_0-T_m)\Gamma\left(-\frac{\alpha}{2}+1\right)M_{\frac{\alpha}{2}}(\xi)}\\
\label{k-ck}&k=\frac{q_0^2\,\Gamma^3\left(-\frac{\alpha}{2}+1\right)\left[1-W\left(-\xi,-\frac{\alpha}{2},1\right)\right]M_{\frac{\alpha}{2}}(\xi)}{\nu_\alpha\rho l(T_0-T_m)\Gamma\left(\frac{\alpha}{2}+1\right)\xi}
\end{align}
\end{subequations}
respectively, $\xi$ being the only one solution to the equation:
\begin{equation}\label{EqXi-ck}
M_{\frac{\alpha}{2}}(x)=\frac{\nu_\alpha\sigma\rho l\Gamma\left(\frac{\alpha}{2}+1\right)}{q_0\Gamma^2\left(-\frac{\alpha}{2}+1\right)},\quad x>0,
\end{equation}
if and only if the following inequality holds:
\begin{equation}\label{Cond-ck}
\frac{\nu_\alpha\sigma\rho l\Gamma\left(\frac{\alpha}{2}+1\right)}{q_0\Gamma\left(-\frac{\alpha}{2}+1\right)}<1.
\end{equation}
\end{theorem}

\begin{proof}
By following the same steps as in the demonstration of the Theorem \ref{Thlc}, it can be shown that $c$ and $k$ must be given by (\ref{ck}), where the parameter $\xi$ should be a solution to equation (\ref{EqXi-ck}). Since the Mainardi function $M_{\frac{\alpha}{2}}$ is a strictly decreasing function from $\frac{1}{\Gamma\left(-\frac{\alpha}{2}+1\right)}$ to $0$ in $\mathbb{R}^+$ \cite{RoSa2013}, we have that the equation (\ref{EqXi-ck}) admits a solution if and only if its RHS is between $0$ and $\frac{1}{\Gamma\left(-\frac{\alpha}{2}+1\right)}$. This is equivalent to say that inequality (\ref{Cond-ck}) must holds. Moreover, when data satisfy (\ref{Cond-ck}), equation (\ref{EqXi-ck}) has an only one positive solution.
\end{proof}

\begin{theorem}[Case 3: Determination of $l$ and $k$]\label{Thlk}
If the moving boundary $s$ is given by (\ref{s}), then the Problem (\ref{Problema}) admits the solution $T$, $k$ and $l$ given by (\ref{T}) and:
\begin{subequations}\label{lk}
\begin{align}
\label{k-lk}&k=\frac{1}{\rho c}\left[\frac{q_0\sqrt{\mu_\alpha}\Gamma\left(-\frac{\alpha}{2}+1\right)\left[1-W\left(-\xi,-\frac{\alpha}{2},1\right)\right]}{T_0-T_m}\right]^2\\
\label{l-lk}&l=\frac{c(T_0-T_m)\Gamma\left(-\frac{\alpha}{2}+1\right)M_{\frac{\alpha}{2}}(\xi)}{\mu_\alpha\nu_\alpha\Gamma\left(\frac{\alpha}{2}+1\right)\xi\left[1-W\left(-\xi,-\frac{\alpha}{2},1\right)\right]}
\end{align}
\end{subequations}
respectively, $\xi$ being the only one solution to the equation:
\begin{equation}\label{EqXi-lk}
G_\alpha(x)=\frac{\sigma\rho c(T_0-T_m)}{\mu_\alpha q_0\Gamma\left(-\frac{\alpha}{2}+1\right)},\quad x>0.
\end{equation}
\end{theorem}

\begin{proof}
By proceeding analogously to the proofs of the previous Theorems \ref{Thlc} and \ref{Thck}, we have that $k$ and $l$ should be given by (\ref{lk}), $\xi$ being a solution to equation (\ref{EqXi-lk}). Since the RHS of the equation (\ref{EqXi-lk}) is a positive number, it follows from (\ref{GProp}) that the equation (\ref{EqXi-lk}) admits an only one solution for any set of data. 
\end{proof}

\begin{theorem}[Case 4: Determination of $c$ and $\rho$]\label{Thcrho}
If the moving boundary $s$ is given by (\ref{s}), then the Problem (\ref{Problema}) admits the solution $T$, $c$ and $\rho$ given by (\ref{T}), (\ref{c-ck}) and:
\begin{subequations}\label{crho}
\begin{align}
\label{rho-crho}&\rho=\frac{q_0^2\,\Gamma^3\left(-\frac{\alpha}{2}+1\right)\left[1-W\left(-\xi,-\frac{\alpha}{2},1\right)\right]M_{\frac{\alpha}{2}}(\xi)}{\nu_\alpha k l(T_0-T_m)\Gamma\left(\frac{\alpha}{2}+1\right)\xi}
\end{align}
\end{subequations}
respectively, $\xi$ being the only one solution to the equation (\ref{EqXi-lc}), if and only if inequality (\ref{Cond-lc}) holds.
\end{theorem}

\begin{proof}
It is similar to the demonstration of Theorem \ref{Thlc}. 
\end{proof}

\begin{theorem}[Case 5: Determination of $l$ and $\rho$]\label{Thlrho}
If the moving boundary $s$ is given by (\ref{s}), then the Problem (\ref{Problema}) admits the solution $T$, $l$ and $\rho$ given by (\ref{T}), (\ref{l-lk}) and:
\begin{subequations}\label{lrho}
\begin{align}
\label{rho-lrho}&\rho=\frac{1}{k c}\left[\frac{q_0\sqrt{\mu_\alpha}\Gamma\left(-\frac{\alpha}{2}+1\right)\left[1-W\left(-\xi,-\frac{\alpha}{2},1\right)\right]}{T_0-T_m}\right]^2
\end{align}
\end{subequations}
respectively, $\xi$ being the only one solution to the equation (\ref{EqXi-lc}), if and only if inequality (\ref{Cond-lc}) holds.
\end{theorem}

\begin{proof}
It is similar to the demonstration of Theorem \ref{Thlc}. 
\end{proof}

\begin{theorem}[Case 6: Determination of $\rho$ and $k$]\label{Thrhok}
If the moving boundary $s$ is given by (\ref{s}), then the Problem (\ref{Problema}) admits the solution $T$, $\rho$ and $k$ given by (\ref{T}) and:
\begin{subequations}\label{rhok}
\begin{align}
\label{rho-rhok}&\rho=\frac{q_0\mu_\alpha\Gamma\left(-\frac{\alpha}{2}+1\right)\xi\left[1-W\left(-\xi,-\frac{\alpha}{2},1\right)\right]}{\sigma c(T_0-T_m)}\\
\label{k-rhok}&k=\frac{\sigma q_0\Gamma\left(-\frac{\alpha}{2}+1\right)\left[1-W\left(-\xi,-\frac{\alpha}{2},1\right)\right]}{(T_0-T_m)\xi}
\end{align}
\end{subequations}
respectively, $\xi$ being the only one solution to the equation:
\begin{equation}\label{EqXi-rhok}
H_\alpha(x)=\frac{c(T_0-T_m)\Gamma\left(-\frac{\alpha}{2}+1\right)}{\mu_\alpha\nu_\alpha l\Gamma\left(\frac{\alpha}{2}+1\right)},\quad x>0.
\end{equation}
\end{theorem}

\begin{proof}
By following the same ideas as in the proof of the Theorem \ref{Thlc}, we have that $\rho$ and $k$ must be given by (\ref{rhok}), where $\xi$ should be a solution to equation (\ref{EqXi-rhok}). By noting that the RHS of this equation is a positive number, it follows from (\ref{HProp}) that the equation (\ref{EqXi-rhok}) admits an only one positive solution for any set of data. 
\end{proof}

Table \ref{TbFrac} summarizes the formulae obtained for the two unknown thermal coefficients and the condition that data must verify to obtain them, for each one of the six possible choices of the two unknown thermal coefficients among $k$, $\rho$, $l$ and $c$ in Problem (\ref{Problema}) when the moving boundary $s$ is defined by (\ref{s}).

\newpage
\begin{table}[h!]
\caption{Formulae for the two unknown thermal coefficients and condition on data for\\ Problem (\ref{Problema})}\label{TbFrac}
\begin{tabular}{cllc}
\hline\\[-0.25cm]
Case & \hspace*{0.5cm}Formulae for the unknown & \hspace*{0.5cm}Equation for $\xi$ & Condition\\
&\hspace*{0.5cm}thermal coefficients& & for data\\[0.25cm]
\hline\\[-0.25cm]
1&
$c=\frac{1}{\rho k}\left[\frac{q_0\sqrt{\mu_\alpha}\Gamma\left(-\frac{\alpha}{2}+1\right)\left[1-W\left(-\xi,-\frac{\alpha}{2},1\right)\right]}{T_0-T_m}\right]^2$&
$F_\alpha(x)=\frac{k(T_0-T_m)}{\sigma q_0\Gamma\left(-\frac{\alpha}{2}+1\right)},\,\, x>0$&
$\frac{k(T_0-T_m)}{\sigma q_0}<1$\\[0.5cm]
& 
$l=\frac{q_0^2\Gamma^3\left(-\frac{\alpha}{2}+1\right)\left[1-W\left(-\xi,-\frac{\alpha}{2},1\right)\right]M_{\frac{\alpha}{2}}(\xi)}{\nu_\alpha\rho k(T_0-T_m)\Gamma\left(\frac{\alpha}{2}+1\right)\xi}$&\\[0.25cm]
\hline\\[-0.25cm]
2&
$c=\frac{\mu_\alpha\nu_\alpha l\Gamma\left(\frac{\alpha}{2}+1\right)\xi\left[1-W\left(-\xi,-\frac{\alpha}{2},1\right)\right]}{(T_0-T_m)\Gamma\left(-\frac{\alpha}{2}+1\right)M_{\alpha/2}(\xi)}$&
$M_{\frac{\alpha}{2}}(x)=\frac{\nu_\alpha\sigma\rho l\Gamma\left(\frac{\alpha}{2}+1\right)}{q_0\Gamma^2\left(-\frac{\alpha}{2}+1\right)},\,\, x>0$&
$\frac{\nu_\alpha\sigma\rho l\Gamma\left(\frac{\alpha}{2}+1\right)}{q_0\Gamma\left(-\frac{\alpha}{2}+1\right)}<1$\\[0.5cm]
& 
$k=\frac{q_0^2\Gamma^3\left(-\frac{\alpha}{2}+1\right)\left[1-W\left(-\xi,-\frac{\alpha}{2},1\right)\right]M_{\frac{\alpha}{2}}(\xi)}{\nu_\alpha\rho l(T_0-T_m)\Gamma\left(\frac{\alpha}{2}+1\right)\xi}$&\\[0.25cm]
\hline\\[-0.25cm]
3&
$k=\frac{1}{\rho c}\left[\frac{q_0\sqrt{\mu_\alpha}\Gamma\left(-\frac{\alpha}{2}+1\right)\left[1-W\left(-\xi,-\frac{\alpha}{2},1\right)\right]}{T_0-T_m}\right]^2$&
$G_\alpha(x)=\frac{\sigma\rho c(T_0-T_m)}{\mu_\alpha q_0\Gamma\left(-\frac{\alpha}{2}+1\right)},\,\, x>0$&
$-$\\[0.5cm]
& 
$l=\frac{c(T_0-T_m)\Gamma\left(-\frac{\alpha}{2}+1\right)M_{\frac{\alpha}{2}}(\xi)}{\mu_\alpha\nu_\alpha\Gamma\left(\frac{\alpha}{2}+1\right)\xi\left[1-W\left(-\xi,-\frac{\alpha}{2},1\right)\right]}$&\\[0.25cm]
\hline\\[-0.25cm]
4&
$c=\frac{\mu_\alpha\nu_\alpha l\Gamma\left(\frac{\alpha}{2}+1\right)\xi\left[1-W\left(-\xi,-\frac{\alpha}{2},1\right)\right]}{(T_0-T_m)\Gamma\left(-\frac{\alpha}{2}+1\right)M_{\alpha/2}(\xi)}$&
$F_\alpha(x)=\frac{k(T_0-T_m)}{\sigma q_0\Gamma\left(-\frac{\alpha}{2}+1\right)},\,\, x>0$&
$\frac{k(T_0-T_m)}{\sigma q_0}<1$\\[0.5cm]
& 
$\rho=\frac{q_0^2\Gamma^3\left(-\frac{\alpha}{2}+1\right)\left[1-W\left(-\xi,-\frac{\alpha}{2},1\right)\right]M_{\frac{\alpha}{2}}(\xi)}{\nu_\alpha k l(T_0-T_m)\Gamma\left(\frac{\alpha}{2}+1\right)\xi}$&\\[0.25cm]
\hline\\[-0.25cm]
5&
$l=\frac{c(T_0-T_m)\Gamma\left(-\frac{\alpha}{2}+1\right)M_{\alpha/2}(\xi)}{\mu_\alpha\nu_\alpha\Gamma\left(\frac{\alpha}{2}+1\right)\xi\left[1-W\left(-\xi,-\frac{\alpha}{2},1\right)\right]}$&
$F_\alpha(x)=\frac{k(T_0-T_m)}{\sigma q_0\Gamma\left(-\frac{\alpha}{2}+1\right)},\,\, x>0$&
$\frac{k(T_0-T_m)}{\sigma q_0}<1$\\[0.5cm]
& 
$\rho=\frac{1}{k c}\left[\frac{q_0\sqrt{\mu_\alpha}\Gamma\left(-\frac{\alpha}{2}+1\right)\left[1-W\left(-\xi,-\frac{\alpha}{2},1\right)\right]}{T_0-T_m}\right]^2$&\\[0.25cm]
\hline\\[-0.25cm]
6&
$\rho=\frac{\mu_\alpha q_0\Gamma\left(-\frac{\alpha}{2}+1\right)\xi\left[1-W\left(-\xi,-\frac{\alpha}{2},1\right)\right]}{\sigma c(T_0-T_m)}$&
$H_\alpha(x)=\frac{c(T_0-T_m)\Gamma\left(-\frac{\alpha}{2}+1\right)}{\mu_\alpha\nu_\alpha\Gamma\left(\frac{\alpha}{2}+1\right)},\,\, x>0$&
$-$\\[0.5cm]
& 
$k=\frac{\sigma q_0\Gamma\left(-\frac{\alpha}{2}+1\right)\left[1-W\left(-\xi,-\frac{\alpha}{2},1\right)\right]}{(T_0-T_m)\xi}$&\\[0.25cm]
\hline
\end{tabular}
\end{table}

\newpage
\section{Convergence to the classic case when $\alpha\to 1^-$}
When $\alpha=1$, the time fractional derivative of order $\alpha$ in the Caputo sense of a function coincides with its classical time derivative. Then, if we allow $\alpha$ to be equal to 1 in Problem (\ref{Problema}) and we consider that case, we obtain that Problem (\ref{Problema}) reduces to the classical inverse one-phase Stefan problem studied in \cite{Ta1983}. This problem, which we will refer to as Problem (\ref{Problema}$^\star$), is given by the classical diffusion equation:
\begin{equation}\label{ED}
T_t(x,t)=\lambda^2T_{xx}(x,t),\quad 0<x<s(t),\,t>0,\tag{\ref{1}$^\star$}
\end{equation} 
the classical Stefan condition:
\begin{equation}\label{CondStefan}
-kT_x(s(t),t)=\rho l\dot{s}(t),\quad t>0\tag{\ref{3}$^\star$}
\end{equation}
and conditions (\ref{2}), (\ref{4}) and (\ref{5}). Of course, to obtain (\ref{ED}) and (\ref{CondStefan}) we have also considered $\mu_\alpha=1$ and $\nu_\alpha=1$ in (\ref{1}) and (\ref{3}), respectively. According to the physical interpretation given in \cite{VoFaGa2013}, as $\alpha$ increases its value to 1, the ability of memory retention of the system diminishes to the limit case of no memory retention corresponding to $\alpha=1$. In this context, classical Stefan problems are representing phase-change processes  whose temporal evolution can be described in terms of {\em local in time} properties ({\em absence of memory}). 

The determination of two unknown thermal coefficients through a classical inverse one-phase Stefan problem was done in \cite{Ta1983}. In that article, necessary and sufficient conditions on data to obtain existence and uniqueness of solution to Problem (\ref{Problema}$^\star$) are given, together with formulae for the unknown thermal coefficients. In several articles \cite{RoSa2013,RoSa2014,RoTa2014,Ta2015-c} it has been proved the convergence when $\alpha\to 1^-$ of the solution to a fractional Stefan problem with $0<\alpha<1$ to the solution to the associated classical problem obtained by considering $\alpha=1$. Encouraged by those works, we are interested in this section in proving the convergence when $\alpha\to 1^-$ of the results obtained in Section \ref{SeFormulas} to the ones given in \cite{Ta1983}. 

In order to emphasize the dependence on $\alpha$ of the formulae given in Theorems \ref{Thlc} to \ref{Thrhok}, we will mention it here explicitly. For example, if we are analysing the convergence of the solution to Problem (\ref{Problema}) given in Theorem \ref{Thlc}, we will refer to it as $T(x,t,\alpha)$, $l(\alpha)$ and $c(\alpha)$. We will also write $\xi(\alpha)$ to represent the coefficient defined by (\ref{xi}).

The next result will be useful in the following:

\begin{lemma}\label{LeWMProp}
For each $x>0$, the Wright and Mainardi functions are such that:
\begin{subequations}\label{WMProp}
\begin{align}
\label{WProp}&1-W\left(-x,-\frac{\alpha}{2},1\right)\to\erf\left(\frac{x}{2}\right),&\text{when }\alpha &\to  1^-\\
\label{MProp}&M_{\alpha/2}(x)\to \frac{1}{\sqrt{\pi}}\exp\left(-\frac{x^2}{4}\right),&\text{when } \alpha &\to  1^-,
\end{align}
\end{subequations}
where $erf$ is the error function, which is defined by:
\begin{equation*}
\erf(x)=\frac{2}{\sqrt{\pi}}\displaystyle\int_0^x\exp(s^2)ds,\quad x>0.
\end{equation*}
\end{lemma}

\begin{proof}
See \cite{RoSa2013}.
\end{proof}  

\begin{theorem}[Convergence related to Case 1]\label{ThlcClassic}
If inequality (\ref{Cond-lc}) holds, then the solution $T(x,t,\alpha)$, $l(\alpha)$, $c(\alpha)$ to Problem (\ref{Problema}) given in Theorem \ref{Thlc} converges to the solution obtained in \cite{Ta1983}, which is given by:
\begin{subequations}\label{lcClassic}
\begin{align}
\label{T-Classic}&T(x,t)=T_0+\frac{T_0-T_m}{\erf\left(\frac{\sigma^\star}{\lambda}\right)}\erf\left(\frac{x}{2\lambda\sqrt{t}}\right),\quad 0<x<s(t),\,t>0\\
\label{c-lcClassic}&c=\frac{k}{\rho}\left(\frac{\xi^\star}{\sigma^\star}\right)^2\\
\label{l-lcClassic}&l=\frac{q_0\exp({-\xi^\star}^2)}{\rho\sigma^\star}
\end{align}
\end{subequations}
where $\sigma^\star$ is defined by:
\begin{equation}\label{sigmaStar}
\sigma^\star=\frac{\sigma}{2}
\end{equation}
and $\xi^\star$ is the only one solution to the equation:
\begin{equation}\label{EqXi-lcClassic}
\frac{\erf(x)}{x}=\frac{k(T_0-T_m)}{q_0\sigma^\star\sqrt{\pi}},\quad x>0.
\end{equation}
\end{theorem}

\begin{proof}
Taking the limit when $\alpha\to 1^-$ into both sides of the equation (\ref{EqXi-lc}) and using (\ref{MProp}), we obtain the following equation:
\begin{equation}\label{EqXi-lcClassicPrelim}
\frac{\erf(x/2)}{x}=\frac{k(T_0-T_m)}{q_0\sigma \sqrt{\pi}},\quad x>0.
\end{equation}
On one hand, we have that the LHS of the equation (\ref{EqXi-lcClassicPrelim}) defines a strictly decreasing function from $\frac{1}{\sqrt{\pi}}$ to $0$ in $\mathbb{R}^+$. On the other hand, since inequality (\ref{Cond-lc}) holds, the RHS of equation (\ref{Cond-lc}) is between $0$ and $\frac{1}{\sqrt{\pi}}$. Therefore, it follows that equation (\ref{EqXi-lcClassicPrelim}) has an only one positive solution. By introducing the parameter $\sigma^\star$ defined by (\ref{sigmaStar}), we can rewrite equation (\ref{EqXi-lcClassicPrelim}) as follows:
\begin{equation*}
\frac{\erf(x/2)}{x/2}=\frac{k(T_0-T_m)}{q_0\sigma^\star \sqrt{\pi}},\quad x>0,
\end{equation*}
and see that the solution $\xi(\alpha)$ to the equation (\ref{EqXi-lc}) is such that:
\begin{equation}\label{xi-lcLimit}
\xi(\alpha)\to 2\xi^\star,\quad\quad\text{when }\alpha\to 1^-
\end{equation}
$\xi^\star$ being the only one solution to the equation (\ref{EqXi-lcClassic}). From (\ref{munu-limit}), (\ref{WMProp}) and (\ref{xi-lcLimit}), it follows from elementary computations that:
\begin{subequations}\label{lc-limit}
\begin{align}
&c(\alpha)\to c,\quad\quad\text{when }\alpha\to 1^-\\
&l(\alpha)\to l,\quad\quad\,\,\text{when }\alpha\to 1^-
\end{align}
\end{subequations}
where $c$ and $l$ are given by (\ref{c-lcClassic}) and (\ref{l-lcClassic}), respectively. Finally, it follows from (\ref{mu-limit}), (\ref{WMProp}), (\ref{xi-lcLimit}) and (\ref{lc-limit}) that:
\begin{equation}\label{T-limit}
T(x,t,\alpha)\to T(x,t),\quad\quad\text{when }\alpha\to 1^-,
\end{equation} 
for each pair $(x,t)$ with $0<x<s(t)$ and $t>0$, where $T(x,t)$ is given by (\ref{T-Classic}). We then have proved that the solution to Problem (\ref{Problema}) given in Theorem \ref{Thlc} converges to the solution to Problem (\ref{Problema}$^\star$) given in \cite{Ta1983} when $\alpha\to 1^-$.
\end{proof}

\begin{remark}
We note that inequality (\ref{Cond-lc}) can be written as:
\begin{equation}\label{Cond-ckClassic}
\frac{k(T_0-T_m)}{2\sigma^\star q_0}<1,
\end{equation}
$\sigma^\star$ being the parameter defined by (\ref{sigmaStar}), which is the condition established in \cite{Ta1983} to ensure the existence and uniqueness of the solution to Problem (\ref{Problema}$^\star$) given by (\ref{lcClassic}).
\end{remark}

\begin{theorem}[Convergence related to Case 2]\label{ThckClassic}
If inequality (\ref{Cond-ck}) holds for each $\alpha\in(0,1)$, then the solution $T(x,t,\alpha)$, $c(\alpha)$, $k(\alpha)$ to Problem (\ref{Problema}) given in Theorem \ref{Thck} converges to the solution obtained in \cite{Ta1983}, which is given by (\ref{T-Classic}) and:
\begin{subequations}\label{ckClassic}
\begin{align}
\label{c-ckClassic}&c=\frac{q_0\sqrt{\pi}\xi^\star\erf(\xi^\star)}{\rho\sigma^\star(T_0-T_m)}\\
\label{k-ckClassic}&k=\frac{\sigma^\star q_0\sqrt{\pi}\erf(\xi^\star)}{(T_0-T_m)\xi^\star}
\end{align}
\end{subequations}
where $\sigma^\star$ is given by (\ref{sigmaStar}) and $\xi^\star$ is the only one solution to the equation:
\begin{equation}\label{EqXi-ckClassic}
\exp(x^2)=\frac{q_0}{\rho l\sigma^\star},\quad x>0.
\end{equation}
\end{theorem}

\begin{proof}
If we take the limit when $\alpha\to 1^-$ side by side of equation (\ref{EqXi-ck}) and we have into consideration (\ref{nu-limit}) and (\ref{MProp}), the following equation is obtained:
\begin{equation}\label{EqXi-ckClassicPrelim}
\exp\left(\frac{x^2}{4}\right)=\frac{2q_0}{\sigma\rho l},\quad x>0.
\end{equation}
Since inequality (\ref{Cond-ck}) holds for all $\alpha\in(0,1)$, we have that the following inequality also holds:
\begin{equation}\label{Cond-ckClassicPrelim}
\frac{2q_0}{\sigma\rho l}< 1.
\end{equation}
Therefore, equation (\ref{EqXi-ckClassicPrelim}) admits only one positive solution. We note that equation (\ref{EqXi-ckClassicPrelim}) can be rewritten as:
\begin{equation*}
\exp\left(\left(\frac{x}{2}\right)^2\right)=\frac{q_0}{\sigma^\star\rho l},\quad x>0,
\end{equation*}
where $\sigma^\star$ is given by (\ref{sigmaStar}), from which we can see that the solution $\xi(\alpha)$ to equation (\ref{EqXi-ck}) is such that:
\begin{equation}\label{xi-ckLimit}
\xi(\alpha)\to 2\xi^\star,\quad\quad\text{when }\alpha\to 1^-,
\end{equation}
$\xi^\star$ being the only one solution to equation (\ref{EqXi-ckClassic}). It follows now from (\ref{munu-limit}) (\ref{WMProp}), (\ref{xi-ckLimit}) and elementary computations that:
\begin{subequations}\label{ck-limit}
\begin{align}
&c(\alpha)\to c\\
&k(\alpha)\to k
\end{align}
\end{subequations}
where $c$ and $k$ are given by (\ref{ck}), respectively. Finally, we have from (\ref{mu-limit}), (\ref{WMProp}), (\ref{xi-ckLimit}) and (\ref{ck-limit}) that $T(x,t,\alpha)$ satisfies (\ref{T-limit}).
\end{proof}

\begin{remark}
By introducing the parameter $\sigma^\star$ defined by (\ref{sigmaStar}), we have that the inequality (\ref{Cond-ckClassicPrelim}) can be rewritten as:
\begin{equation}\label{Cond-ckClassic}
\frac{q_0}{\rho l\sigma^\star}>1,
\end{equation}
which is the condition established in \cite{Ta1983} to ensure the existence and uniqueness of the solution to Problem (\ref{Problema}$^\star$) given by (\ref{T-Classic}) and (\ref{ckClassic}).
\end{remark}

\begin{theorem}[Convergence related to Case 3]\label{ThlkClassic}
If inequality (\ref{Cond-lc}) holds, then the solution $T(x,t,\alpha)$, $l(\alpha)$, $k(\alpha)$ to Problem (\ref{Problema}) given in Theorem \ref{Thlk} converges to the solution obtained in \cite{Ta1983}, which is given by (\ref{T-Classic}) and:
\begin{subequations}\label{lkClassic}
\begin{align}
\label{l-lkClassic}&l=\frac{q_0\exp(-{\xi^\star}^2)}{\rho\sigma^\star}\\
\label{k-lkClassic}&k=\rho c\left(\frac{\sigma^\star}{\xi^\star}\right)^2
\end{align}
\end{subequations}
where $\sigma^\star$ is given by (\ref{sigmaStar}) and $\xi^\star$ is the only one solution to the equation:
\begin{equation}\label{EqXi-lkClassic}
x\erf(x)=\frac{\rho c\sigma^\star(T_0-T_m)}{q_0\sqrt{\pi}},\quad x>0.
\end{equation}
\end{theorem}

\begin{proof}
Taking the limit when $\alpha\to 1^-$ into both sides of the equation (\ref{EqXi-lk}) and having into consideration (\ref{mu-limit}) and (\ref{WProp}), we obtain the following equation:
\begin{equation}\label{EqXi-lkClassicPrelim}
x\erf\left(\frac{x}{2}\right)=\frac{\sigma\rho c(T_0-T_m)}{q_0\sqrt{\pi}},\quad x>0.
\end{equation}
Since the LHS of equation (\ref{EqXi-lkClassicPrelim}) defines a strictly increasing function from $0$ to $+\infty$ in $\mathbb{R}^+$ and the RHS of equation (\ref{EqXi-lkClassicPrelim}) is a positive number, it follows that equation (\ref{EqXi-lkClassicPrelim}) has an only one positive solution.  
By introducing the parameter $\sigma^\star$ defined by (\ref{sigmaStar}), the equation (\ref{EqXi-lkClassicPrelim}) can be rewritten as:
\begin{equation*}
\frac{x}{2}\erf\left(\frac{x}{2}\right)=\frac{\sigma^\star\rho c(T_0-T_m)}{q_0\sqrt{\pi}},\quad x>0,
\end{equation*}
from which we can see that the solution $\xi(\alpha)$ to the equation (\ref{EqXi-lk}) is such that:
\begin{equation}\label{xi-lkLimit}
\xi(\alpha)\to 2\xi^\star,\quad\quad\text{when }\alpha\to 1^-
\end{equation}
$\xi^\star$ being the only one solution to equation (\ref{EqXi-lkClassic}). The rest of the proof runs as before. 
\end{proof}

The following results can be proved in the same manner than the previous theorems in this section. Then, we have prefered to omit their proofs.

\begin{theorem}[Convergence related to Case 4]\label{ThcrhoClassic}
If inequality (\ref{Cond-lc}) holds, then the solution $T(x,t,\alpha)$, $c(\alpha)$, $\rho(\alpha)$ to Problem (\ref{Problema}) given in Theorem \ref{Thcrho} converges to the solution obtained in \cite{Ta1983}, which is given by (\ref{T-Classic}) and:
\begin{subequations}\label{crhoClassic}
\begin{align}
\label{c-crhoClassic}&c=\frac{kl{\xi^\star}^2\exp({\xi^\star}^2)}{q_0\sqrt{\pi}}\\
\label{rho-crhoClassic}&\rho=\frac{q_0\exp(-{\xi^\star}^2)}{l\sigma^\star}
\end{align}
\end{subequations}
where $\sigma^\star$ is given by (\ref{sigmaStar}) and $\xi^\star$ is the only one solution to the equation (\ref{EqXi-lcClassic}).
\end{theorem}

\begin{theorem}[Convergence related to Case 5]\label{ThlrhoClassic}
The solution $T(x,t,\alpha)$, $l(\alpha)$, $\rho(\alpha)$ to Problem (\ref{Problema}) given in Theorem \ref{Thlrho} converges to the solution obtained in \cite{Ta1983}, which is given by (\ref{T-Classic}) and:
\begin{subequations}\label{lrhoClassic}
\begin{align}
\label{l-lrhoClassic}&l=\frac{q_0 c\sigma^\star\exp(-{\xi^\star}^2)}{k{\xi^\star}^2}\\
\label{rho-lrhoClassic}&\rho=\frac{k}{c}\left(\frac{\xi^\star}{\sigma^\star}\right)^2
\end{align}
\end{subequations}
where $\sigma^\star$ is given by (\ref{sigmaStar}) and $\xi^\star$ is the only one solution to the equation (\ref{EqXi-lcClassic}).
\end{theorem}

\begin{theorem}[Convergence related to Case 6]\label{ThrhokClassic}
The solution $T(x,t,\alpha)$, $\rho(\alpha)$, $k(\alpha)$ to Problem (\ref{Problema}) given in Theorem \ref{Thrhok} converges to the solution obtained in \cite{Ta1983}, which is given by (\ref{T-Classic}) and:
\begin{subequations}\label{rhokClassic}
\begin{align}
\label{rho-rhokClassic}&\rho=\frac{q_0\exp(-{\xi^\star}^2)}{l\sigma^\star}\\
\label{k-rhokClassic}&k=\frac{q_0c\sigma^\star\exp(-{\xi^\star}^2)}{l{\xi^\star}^2}
\end{align}
\end{subequations}
where $\sigma^\star$ is given by (\ref{sigmaStar}) and $\xi^\star$ is the only one solution to the equation:
\begin{equation}\label{EqXi-rhokClassic}
x\erf(x)\exp(x^2)=\frac{c(T_0-T_m)}{l\sqrt{\pi}},\quad x>0.
\end{equation}
\end{theorem}

Table \ref{TbClassic} summarizes the formulae obtained for the two unknown thermal coefficients and the condition that data must verify to obtain them, for each one of the six possible choices for the two unknown thermal coefficients among $k$, $\rho$, $l$ and $c$ in Problem (\ref{Problema}) when $\alpha\to 1^-$.

\newpage
\begin{table}[h!]
\caption{Formulae for the two unknown thermal coefficients and condition on data for\\ Problem (\ref{Problema}) when $\alpha\to 1^-$}\label{TbClassic}
\begin{tabular}{cllc}
\hline\\[-0.25cm]
Case & \hspace*{0.5cm}Formulae for the unknown & \hspace*{0.5cm}Equation for $\xi^\star$ & Condition\\
&\hspace*{0.5cm}thermal coefficients& & for data\\[0.25cm]
\hline\\[-0.25cm]
1&
$c=\frac{k}{\rho}\left(\frac{\xi^\star}{\sigma^\star}\right)^2$&
$\frac{\erf(x)}{x}=\frac{k(T_0-T_m)}{q_0\sigma^\star\sqrt{\pi}},\,\, x>0$&
$\frac{k(T_0-T_m)}{2\sigma^\star q_0}<1$\\[0.5cm]
& 
$l=\frac{q_0\exp(-{\xi^\star}^2)}{\rho\sigma^\star}$&\\[0.25cm]
\hline\\[-0.25cm]
2&
$c=\frac{q_0\sqrt{\pi\xi^\star\erf(\xi^\star)}}{\rho\sigma^\star(T_0-T_m)}$&
$\exp(x^2)=\frac{q_0}{\rho l\sigma^\star},\,\, x>0$&
$\frac{q_0}{\rho l\sigma^\star}$\\[0.5cm]
& 
$k=\frac{\sigma^\star q_0\sqrt{\pi\erf(\xi^2)}}{(T_0-T_m)\xi^\star}$&\\[0.25cm]
\hline\\[-0.25cm]
3&
$k=\rho c\left(\frac{\sigma^\star}{\xi^\star}\right)^2$&
$x\erf(x)=\frac{\rho c\sigma^\star(T_0-T_m)}{q_0\sqrt{\pi}},\,\, x>0$&
$-$\\[0.5cm]
& 
$l=\frac{q_0\exp(-{\xi^\star}^2)}{\rho\sigma^\star}$&\\[0.25cm]
\hline\\[-0.25cm]
4&
$c=\frac{kl{\xi^\star}^2\exp({\xi^\star}^2)}{q_0\sigma^\star}$&
$\frac{\erf(x)}{x}=\frac{k(T_0-T_m)}{q_0\sigma^\star\sqrt{\pi}},\,\, x>0$&
$\frac{k(T_0-T_m)}{2\sigma^\star q_0}<1$\\[0.5cm]
& 
$\rho=\frac{q_0\exp(-{\xi^\star}^2)}{l\sigma^\star}$&\\[0.25cm]
\hline\\[-0.25cm]
5&
$l=\frac{q_0c\sigma^\star\exp(-{\xi^\star}^2)}{k{\xi^\star}^2}$&
$\frac{\erf(x)}{x}=\frac{k(T_0-T_m)}{q_0\sigma^\star\sqrt{\pi}},\,\, x>0$&
$\frac{k(T_0-T_m)}{2\sigma^\star q_0}<1$\\[0.5cm]
& 
$\rho=\frac{k}{\rho}\left(\frac{\xi^\star}{\sigma^\star}\right)^2$&\\[0.25cm]
\hline\\[-0.25cm]
6&
$\rho=\frac{q_0\exp(-{\xi^\star}^2)}{l\sigma^\star}$&
$x\erf(x)\exp(x^2)=\frac{c(T_0-T_m)}{l\sqrt{\pi}},\,\, x>0$&
$-$\\[0.5cm]
& 
$k=\frac{q_0c\sigma^\star\exp(-{\xi^\star}^2)}{l{\xi^\star}^2}$&\\[0.25cm]
\hline
\end{tabular}
\end{table}
\newpage
\section*{Conclusions}
In this article we have considered a semi-infinite one-dimensional phase-change material with two unknown constant thermal coefficients. These were assumed to be among the latent heat per unit mass, the specific heat, the mass density and the thermal conductivity. The determination of them have been done through an inverse one-phase fractional Stefan problem with an overspecified condition at the fixed boundary of the material and a known evolution of the moving boundary. It was considered that this problem corresponds to a melting process with latent-heat memory effects, which we have represented by replacing the classical time derivative involved in the diffusion equation and the Stefan condition, by a time fractional derivative of order $\alpha$ ($0<\alpha<1$) in the Caputo sense. Solutions of similarity type were looked for and necessary and sufficient conditions on data to have their existence and uniqueness were given for each of the six inverse fractional Stefan problems that arise according to the choice of the two unknown thermal coefficients. We have also obtained explicit expressions for the temperature and the two unknown thermal coefficients. Finally, we have compared our results with those obtained for the determination of two coefficients through the classical statement ($\alpha=1$) of the inverse Stefan problem and we have proved the convergence of our results (which takes into account latent-heat memory effects) to those obtained by the classic case (no memory retention).  

\subsection*{Acknowledgements}
This paper has been partially sponsored by the Project PIP No. 0534 from CONICET-UA (Rosario, Argentina) and AFOSR-SOARD Grant FA 9550-14-1-0122.

\bibliographystyle{plain}
\bibliography{References_2016-08-23}

\end{document}